\documentclass[11pt,a4paper]{article}

\usepackage[T1]{fontenc}
\usepackage[utf8]{inputenc}
\usepackage{lmodern}
\usepackage{amsmath,amssymb,mathtools,amsthm}
\usepackage{booktabs}
\usepackage[margin=1.05in]{geometry}
\usepackage[expansion=false]{microtype}
\usepackage[square,numbers,sort&compress]{natbib}
\usepackage[hidelinks]{hyperref}

\theoremstyle{plain}
\newtheorem{theorem}{Theorem}[section]

\newtheorem{lemma}[theorem]{Lemma}
\newtheorem{corollary}[theorem]{Corollary}
\theoremstyle{definition}

\newtheorem{hypothesis}[theorem]{Hypothesis}
\newtheorem{remark}[theorem]{Remark}

\numberwithin{equation}{section}
\newcommand{\E}{\mathbb{E}}

\newcommand{\Lop}{\mathcal{L}}
\newcommand{\vk}{\mathbf{v}_\kappa}
\newcommand{\Pplus}{\mathcal{P}_{+}}
\newcommand{\Cex}{\mathcal{C}}
\newcommand{\M}[1][]{\def\Mtmp{#1}\ifx\Mtmp\empty[M]\else[M,~#1]\fi}

\title{\bfseries Buying at Support:\\ The Entry Problem under Multi-Skew
Brownian Motion}
\author{Jun Maeda\thanks{Any errors are the author's own.}\\[4pt]
\normalsize \texttt{jun.maeda@warwickgrad.net}}
\date{October 2026}

\begin{document}
\maketitle

\begin{abstract}
We study when to buy a share that will later be sold optimally, when the
price follows a geometric multi-skew Brownian motion whose skew levels model
support and resistance. The reward for buying is the exit premium of the
liquidation problem solved in a companion paper. This premium is strictly
$r$-subharmonic inside the exit continuation region, so buying is optimal
only at levels where the price is pushed upward, and the entry problem
reduces to a finite one whose value is an upper concave hull of finitely many
points. With one support and one resistance, the optimal rule is to buy
exactly at the support, in every exit regime, with an explicit value.
\end{abstract}

\section{Introduction}
\label{sec:intro}

The companion paper~\cite{M26}, cited as \M{} below, solves the exit
problem
\[
v(x)=\sup_{\tau}\E_x\bigl[e^{-r\tau}S_\tau\bigr]
\]
for a geometric multi-skew Brownian motion $S$ with local-time pushes at
finitely many levels. It concludes that ``sell into resistance'' is optimal
exactly when a closed-form criterion holds \M[Thm.~4.3]. The folklore slogan
has a second half, \emph{buy at support}. This paper shows that the second
half holds with no qualification at all.

Buying \emph{without} a subsequent sale is not a meaningful problem: with
$\mu<r$, waiting always lowers the expected discounted purchase price
(Remark~\ref{rem:naive}). An investor buys in order to acquire the right to
sell later, so the reward for entering at $x$ is $v(x)-x$. The entry problem
is therefore
\begin{equation}
u(x)=\sup_{\tau}\E_x\bigl[e^{-r\tau}W(S_\tau)\bigr],
\qquad W:=v-g,\quad g(x)=x,
\label{eq:entry}
\end{equation}
with $e^{-r\tau}W(S_\tau):=0$ on $\{\tau=\infty\}$. The function $W$ is the
one whose positivity is established in \M[Prop.~5.1]; there it is only a step
in a verification argument, and nothing is optimised over it.

\paragraph{Results.}
\begin{enumerate}
\item \emph{Entry duality} (Theorem~\ref{thm:duality}). For any finite
  configuration of skew points, the entry region is contained in
  $\Pplus=\{z:\beta_z>0\}$. These are exactly the levels at which, by the
  short-time argument of \citet{AlvarezSalminen17}, the pay-off fails to be
  locally $r$-excessive, so they generate the exit continuation region. One
  buys exactly where the value of waiting to sell comes from.
\item \emph{Finite hull} (same theorem). $u/\varphi_r$ is the upper concave
  hull of $(0,0)$ and the finitely many points
  $\bigl(F(z),W(z)/\varphi_r(z)\bigr)$, $z\in\Pplus$, continued flat to the
  right. The entry region is the set of those $z$ whose point lies on the
  hull. Given the exit solution, the entry problem is finite-dimensional.
\item \emph{Support and resistance} (Corollary~\ref{cor:main}). With
  $\beta_L>0>\beta_H$ the entry region is $\{L\}$ in regimes A, B and C
  alike, and $u$ is explicit. There is no free boundary and no smooth fit;
  $u$ has a concave corner at $L$.
\item \emph{Two supports} (Corollary~\ref{cor:samesign}). With
  $\beta_L,\beta_H>0$ the entry region is $\{L\}$, $\{H\}$ or $\{L,H\}$
  according to where $W(L)/W(H)$ falls relative to the window
  $[\psi_r(L)/\psi_r(H),\,\varphi_r(L)/\varphi_r(H)]$, which always contains
  $1$ in its interior.
\end{enumerate}

\paragraph{Related literature.}
The entry--exit formulation itself is not new. In the path-dependent
regime-switching model of support and resistance, \citet{Henderson26} pose a
buyer's problem as a double stopping problem and reduce it to one of exactly
this form, with reward $g=V-u$ for a utility $u$ \citep[(5.1)--(5.2)]{Henderson26}.
There the buyer never buys in two of the three regimes (apart from the
trivial point $0$ when it is absorbing), and in the third the buying set lies
in $(L,A]$ and, when the decay rate $\lambda_0$ of the $+$ regime vanishes,
is an interval $[a,b]$ with $L<a<b\le A$ \citep[Thm.~5.2]{Henderson26}.
\citet{LeungLi15} use the reward $v-g-c$ for an Ornstein--Uhlenbeck spread and
again obtain an entry interval, with a disconnected waiting region;
\citet{LeungLiWang14} find threshold rules for both a single entry--exit pair
and the corresponding switching problem for a Cox--Ingersoll--Ross process
with fixed costs, together with conditions under which it is optimal never to
enter. For
general one-dimensional diffusions with fixed costs,
\citet{ZervosJohnsonAlazemi13} show that buying in order to sell later is
never optimal when $0$ is a natural boundary, as for geometric Brownian
motion. Our model reduces to geometric Brownian motion when every
$\beta_z=0$, and then indeed $W\equiv0$; the skew is exactly what breaks that
no-trade conclusion.

On the stopping side, \citet{AlvarezSalminen17} show that for skew Brownian
motion the skew point lies in the continuation region for non-decreasing
pay-offs without a concave kink there, and \citet{SalminenTa15} relate the
failure of smooth fit to singularities of the scale and speed measures.
Isolated stopping points at a skew level have been observed before
(Remark~\ref{rem:AS}). What is new here is the shape of the whole answer: in
the memoryless skew model the entry region is never an interval but a finite
set of levels, determined by a finite hull, and in the support--resistance
model it is the single point $L$. We have found no result of the form of
Theorem~\ref{thm:duality} in the literature.

\paragraph{Outline.}
Section~\ref{sec:setting} sets up the model, recalls from \M{} the
It\^o--Tanaka formula at the skew levels, the crossing matrices and the
fundamental pair, and states the properties of the exit solution that the
entry problem uses (Hypothesis~\ref{hyp:exit}). Section~\ref{sec:duality}
proves the entry duality and the finite-hull representation
(Theorem~\ref{thm:duality}) and discusses transaction costs.
Section~\ref{sec:main} specialises to one support and one resistance,
Section~\ref{sec:samesign} to two supports, and Section~\ref{sec:othersigns}
to the remaining sign pairs. Section~\ref{sec:disc} concludes with open
problems, and Appendix~\ref{app:num} reports a numerical check that assumes
none of the results.

\section{Setting}
\label{sec:setting}

\subsection{The process}

Fix skew levels $0<z_1<\dots<z_n$ with coefficients
$\beta_i:=\beta_{z_i}\in(-1,1)\setminus\{0\}$, and let
\[
dS_t=\mu S_t\,dt+\sigma S_t\,dW_t+\sum_{i=1}^n\frac{\beta_i}{z_i}S_t\,
d\ell^{z_i}_t(S),\qquad S_0=x>0,
\]
the $n$-level version of \M[(2)]; it has a unique strong solution, and $S$
is a regular strong Markov diffusion on $(0,\infty)$ \M[Prop.~2.1]. Write
$Z:=\{z_1,\dots,z_n\}$ and
\[
\Pplus:=\{z\in Z:\beta_z>0\}
\]
for the upward-skew (\emph{support-type}) levels. Throughout we assume, as
in \M[Ass.~2.4],
\begin{equation}
0<r\qquad\text{and}\qquad\mu<r .
\label{eq:ass}
\end{equation}

Let $\Lop f:=\tfrac12\sigma^2x^2f''+\mu xf'$ off $Z$, and for $f$ continuous
and piecewise $C^2$ with one-sided derivatives at $z\in Z$ put
\begin{equation}
\Delta_z(f):=(1+\beta_z)f'(z^+)-(1-\beta_z)f'(z^-).
\label{eq:Delta}
\end{equation}

\begin{lemma}[It\^o--Tanaka at the skew levels {\M[Lemma~2.2 and \S2]}]
\label{lem:ito}
Let $f$ be continuous on $(0,\infty)$ and $C^2$ on the closure of each
interval of $(0,\infty)\setminus Z$, with finite one-sided limits of $f'$ and
$f''$ at each $z\in Z$. Then
\[
d\bigl(e^{-rt}f(S_t)\bigr)=e^{-rt}(\Lop-r)f(S_t)\,dt
+\tfrac12e^{-rt}\sum_{z\in Z}\Delta_z(f)\,d\ell^z_t+dN_t
\]
for a local martingale $N$, where $(\Lop-r)f$ is defined off $Z$, a set at
which $S$ spends zero time. Moreover $f$ lies in the domain of the generator
if and only if $\Delta_z(f)=0$ for every $z\in Z$ (the \emph{transmission
condition}).
\end{lemma}

The local-time coefficient arises as follows: the jump of $f'$ at $z$
contributes $\tfrac12\bigl(f'(z^+)-f'(z^-)\bigr)\,d\ell^z$ to the
It\^o--Tanaka--Meyer formula, and the push $(\beta_z/z)S\,d\ell^z$ in the
dynamics contributes $\tfrac12\beta_z\bigl(f'(z^+)+f'(z^-)\bigr)\,d\ell^z$;
their sum is $\tfrac12\Delta_z(f)\,d\ell^z$. Following \M[\S2], we call $f$
\emph{$r$-harmonic} on an interval $I$ if $(\Lop-r)f=0$ on $I\setminus Z$
and $\Delta_z(f)=0$ for $z\in Z\cap I$, and \emph{$r$-superharmonic} if
instead $(\Lop-r)f\le0$ and $\Delta_z(f)\le0$. By Lemma~\ref{lem:ito},
$e^{-rt}f(S_t)$ is then a local martingale, respectively a local
supermartingale, while $S$ stays in $I$; a non-negative function that is
$r$-superharmonic on $(0,\infty)$ is therefore $r$-excessive. For the pay-off
$g(x)=x$,
\begin{equation}
(\Lop-r)g=(\mu-r)x<0,\qquad \Delta_z(g)=2\beta_z .
\label{eq:gfacts}
\end{equation}

\subsection{Exponents, crossing matrices and smooth fit}

The functions $x^\alpha$ solving $(\Lop-r)f=0$ have
$\tfrac12\sigma^2\alpha(\alpha-1)+\mu\alpha-r=0$, whose roots satisfy
$\alpha_1<0<1<\alpha_2$ under \eqref{eq:ass} \M[Lemma~3.1]: the product of
the roots is $-2r/\sigma^2<0$, and the quadratic equals $\mu-r<0$ at
$\alpha=1$. Put
\[
\kappa_1:=\frac{\alpha_2-1}{\alpha_2-\alpha_1},\qquad
\kappa_2:=\frac{1-\alpha_1}{\alpha_2-\alpha_1},
\]
both positive with $\kappa_1+\kappa_2=1$. On each interval of
$(0,\infty)\setminus Z$ an $r$-harmonic function is
$Ax^{\alpha_1}+Bx^{\alpha_2}$, which we identify with its coefficient pair
$c=(A,B)^{\!\top}$; with $e_x:=(x^{\alpha_1},x^{\alpha_2})$ for $x>0$, its
value at $x$ is $e_xc$.

\begin{lemma}[crossing {\M[Lemma~3.2]}]
\label{lem:crossing}
Let $z\in Z$ with coefficient $\beta=\beta_z$, and let $f$ equal
$A^dx^{\alpha_1}+B^dx^{\alpha_2}$ just below $z$ and
$A^ux^{\alpha_1}+B^ux^{\alpha_2}$ just above, with pairs
$c^d=(A^d,B^d)^{\!\top}$ and $c^u=(A^u,B^u)^{\!\top}$. Then $f$ is continuous at
$z$ with $\Delta_z(f)=0$ if and only if $c^{d}=T_\beta(z)\,c^{u}$, where
\[
T_\beta(z)=\frac{1}{(\alpha_2-\alpha_1)(1-\beta)}
\begin{pmatrix}
(\alpha_2-\alpha_1)-(\alpha_1+\alpha_2)\beta & -2\alpha_2\beta z^{\alpha_2-\alpha_1}\\[3pt]
2\alpha_1\beta z^{\alpha_1-\alpha_2} & (1+\beta)\alpha_2-(1-\beta)\alpha_1
\end{pmatrix}.
\]
Moreover $T_\beta(z)^{-1}=T_{-\beta}(z)$.
\end{lemma}

\begin{proof}
Put $a:=A^uz^{\alpha_1}$, $b:=B^uz^{\alpha_2}$, $\tilde a:=A^dz^{\alpha_1}$ and
$\tilde b:=B^dz^{\alpha_2}$, so that $zf'(z^+)=\alpha_1a+\alpha_2b$ and
$zf'(z^-)=\alpha_1\tilde a+\alpha_2\tilde b$. Continuity at $z$ and
$\Delta_z(f)=0$ read
\[
\tilde a+\tilde b=a+b,\qquad
(1-\beta)(\alpha_1\tilde a+\alpha_2\tilde b)=(1+\beta)(\alpha_1a+\alpha_2b).
\]
Given $(a,b)$, this is a linear system for $(\tilde a,\tilde b)$ with
determinant $(1-\beta)(\alpha_2-\alpha_1)\ne0$, so it has exactly one
solution; solving it gives $c^d=T_\beta(z)c^u$ \M[Lemma~3.2]. Hence the two
conditions hold if and only if $c^d=T_\beta(z)c^u$. The identity
$T_\beta^{-1}=T_{-\beta}$ is also proved in \M[Lemma~3.2]: the transmission
condition with parameter $\beta$, read from above to below, is the one with
parameter $-\beta$ read from below to above.
\end{proof}

\begin{lemma}[smooth fit {\M[Lemma~3.3]}]
\label{lem:smoothfit}
If $f=Ax^{\alpha_1}+Bx^{\alpha_2}$ satisfies $f(m)=m$ and $f'(m)=1$ for some
$m>0$, then $(A,B)^{\!\top}=\vk(m):=\bigl(\kappa_1m^{1-\alpha_1},\,
\kappa_2m^{1-\alpha_2}\bigr)^{\!\top}$.
\end{lemma}

\begin{proof}
With $a=Am^{\alpha_1}$ and $b=Bm^{\alpha_2}$ the two conditions read $a+b=m$
and $\alpha_1a+\alpha_2b=m$, whence $a=\kappa_1m$ and $b=\kappa_2m$.
\end{proof}

For a level $z>0$ we write
\begin{equation}
G_z(m):=e_z\vk(m)=\kappa_1z^{\alpha_1}m^{1-\alpha_1}+\kappa_2z^{\alpha_2}m^{1-\alpha_2}
\label{eq:Gz}
\end{equation}
for the value at $z$ of the $r$-harmonic piece anchored by smooth fit at $m$,
provided no skew level lies strictly between $m$ and $z$. In the notation of
\M{}, $G_L$ is the function $G$ of \M[Thm.~4.1] and $G_H$ is that of
\M[(26)].

\subsection{The fundamental pair}

Let $\psi_r$ be the $r$-harmonic function on $(0,\infty)$ with
$\psi_r=x^{\alpha_2}$ on $(0,z_1]$, continued upward through each $z_i$ by
$T_{-\beta_i}(z_i)$ (Lemma~\ref{lem:crossing}), and let $\varphi_r$ be the one
with $\varphi_r=x^{\alpha_1}$ on $[z_n,\infty)$, continued downward by
$T_{\beta_i}(z_i)$. These are the increasing and decreasing fundamental
solutions, as in \M[Def.~3.5]. The proof of \M[Lemma~3.6] is a pointwise
induction across the skew points and carries over verbatim to $n$ levels:
$\psi_r>0$, $\psi_r'>0$, $\varphi_r>0$, $\varphi_r'<0$, and the Wronskian
\begin{equation}
\mathcal W:=\varphi_r\psi_r'-\varphi_r'\psi_r>0
\label{eq:wronskian}
\end{equation}
everywhere, with one-sided derivatives at skew points. Put
\[
F:=\psi_r/\varphi_r .
\]
Then $F$ is continuous and strictly increasing, with $F(0^+)=0$ (on
$(0,z_1]$, $\psi_r=x^{\alpha_2}\to0$ while $\varphi_r\ge\varphi_r(z_1)>0$)
and $F(\infty)=\infty$ (symmetrically). Every $r$-harmonic function on an
interval is of the form $A\varphi_r+B\psi_r$ there, and
$(A\varphi_r+B\psi_r)/\varphi_r=A+BF$ is affine in $F$.

\begin{lemma}[hitting times]
\label{lem:hitting}
For $p>0$ let $\tau_p:=\inf\{t\ge0:S_t=p\}$. Then
$\E_x[e^{-r\tau_p}]=\psi_r(x)/\psi_r(p)$ for $x\le p$ and
$\E_x[e^{-r\tau_p}]=\varphi_r(x)/\varphi_r(p)$ for $x\ge p$.
\end{lemma}

\begin{proof}
Let $x\le p$ and $\tau:=\tau_p$. Since $\psi_r$ is $r$-harmonic, Lemma~\ref{lem:ito}
gives
$d\bigl(e^{-rt}\psi_r(S_t)\bigr)=dN_t$ with $N$ a local martingale: the
$dt$-term vanishes off the Lebesgue-null set of times at which $S\in Z$, and
the local-time terms vanish because $\Delta_z(\psi_r)=0$ for every
$z\in Z$, $\psi_r$ having been continued across each skew level by the
crossing matrices (Lemma~\ref{lem:crossing}). Hence
$M_t:=e^{-r(t\wedge\tau)}\psi_r(S_{t\wedge\tau})$ is a local martingale.
By continuity $S_s\le p$ for $s\le\tau$, and $\psi_r$ is increasing, so
$0<M_t\le\psi_r(p)$; a bounded local martingale is a martingale, whence
\[
\psi_r(x)=\E_x[M_t]
=\psi_r(p)\,\E_x\bigl[e^{-r\tau};\tau\le t\bigr]
+\E_x\bigl[e^{-rt}\psi_r(S_t);\tau>t\bigr].
\]
As $t\to\infty$ the first term increases to
$\psi_r(p)\E_x[e^{-r\tau}\mathbf 1_{\{\tau<\infty\}}]$ by monotone
convergence, and the second is at most $\psi_r(p)e^{-rt}\to0$ since $r>0$.
Since $r>0$, we set $e^{-r\tau}:=0$ on $\{\tau=\infty\}$, the limit of
$e^{-rt}$; then $e^{-r\tau}=e^{-r\tau}\mathbf 1_{\{\tau<\infty\}}$ pointwise,
and dividing by $\psi_r(p)>0$ gives the claim. For $x\ge p$ the path stays in $[p,\infty)$ until $\tau$, where $\psi_r$ is
unbounded, so we use $\varphi_r$ instead: it is $r$-harmonic and decreasing,
hence $0<e^{-r(t\wedge\tau)}\varphi_r(S_{t\wedge\tau})\le\varphi_r(p)$, and the
same argument gives $\varphi_r(x)=\varphi_r(p)\,\E_x[e^{-r\tau}]$.
\end{proof}

\subsection{The exit solution as input}

\begin{hypothesis}
\label{hyp:exit}
The exit value $v$ of \M[(6)] is known and has the following form. The exit
continuation region $\Cex:=\{v>g\}$ is a finite union of open intervals, each
with compact closure in $(0,\infty)$. Further, $v$ is continuous, $v=g$ off
$\Cex$, and on each component $J$ of $\Cex$, $v$ is $r$-harmonic in the
sense of Section~\ref{sec:setting}; that is, $v$ is $C^2$ on the closure of each
skew-free piece of $J$ and $\Delta_z(v)=0$ for $z\in Z\cap J$.
\end{hypothesis}

Hypothesis~\ref{hyp:exit} holds in every case that \M{} solves: regimes A,
B, C (Theorems~4.1, 4.3 and 5.2 there), the same-sign case
(Section~7 there, Theorem~7.10), and the remaining sign pairs
(Propositions~8.1 and 8.2 there, with $\Cex=\emptyset$ in the former). It
allows a boundary point of $\Cex$ to be a skew point at which smooth fit
fails, as in regime~C, since only continuity of $v$ is required there.

\begin{remark}[the naive entry problem is degenerate]
\label{rem:naive}
The problem $\inf_\tau\E_x[e^{-r\tau}S_\tau]$ equals $0$ and is attained
by $\tau=\infty$, since the pay-off is non-negative and never buying costs
nothing. Nor is the degeneracy an artefact of allowing $\tau=\infty$: off
$Z$, $e^{-rt}S_t$ is a strict supermartingale by \eqref{eq:gfacts}, so
delaying a purchase that must eventually be made lowers its expected
discounted price. This mirrors the fact that without local time the holder
sells at once \M[\S2.2]. The well-posed entry
problem is \eqref{eq:entry}.
\end{remark}

\begin{remark}[short positions are degenerate]
\label{rem:short}
The mirror round trip, opening a short position and covering it later, is
degenerate under \eqref{eq:ass}. Covering at $\tau$ costs
$\E_x[e^{-r\tau}S_\tau]$, and by the argument of Remark~\ref{rem:naive}
waiting always lowers this cost, so its infimum is $0$: with the convention
that the pay-off vanishes on $\{\tau=\infty\}$ the short seller never
covers, and if covering in finite time is required the infimum is
approached but not attained, so no optimal covering time exists. A
non-trivial short problem needs $\mu>r$, which reverses the standing
assumption and so changes the setting; we do not consider it here.
\end{remark}

\section{Entry duality and the finite hull}
\label{sec:duality}

\begin{lemma}[local structure of the exit premium]
\label{lem:Wlocal}
Under Hypothesis~\ref{hyp:exit}, $W=v-g$ satisfies:
\begin{enumerate}
\item[(i)] $W$ is continuous, non-negative and compactly supported in
  $(0,\infty)$, with $W>0$ exactly on $\Cex$;
\item[(ii)] on $\Cex\setminus Z$, $W$ is $C^2$ and
  $(\Lop-r)W=(r-\mu)x>0$;
\item[(iii)] for $z\in Z\cap\Cex$, $\Delta_z(W)=-2\beta_z$;
\item[(iv)] $\Pplus\subset\Cex$.
\end{enumerate}
\end{lemma}

\begin{proof}
(i) is immediate from Hypothesis~\ref{hyp:exit} and $v\ge g$.
(ii) and (iii) follow by subtracting \eqref{eq:gfacts} from $(\Lop-r)v=0$
and $\Delta_z(v)=0$. (Part (ii) is \M[(28)], read as a statement about
$W$ rather than as a step in a proof.) For (iv), let $z\in\Pplus$. The
short-time expansion \M[(21)] gives
$\E_z[e^{-rt}g(S_t)]-g(z)=\beta_z\E_z[\ell^z_t]+O(t)$, with
$\E_z[\ell^z_t]\asymp\sqrt t$. This is strictly positive for small $t$, so
$v(z)>g(z)$.
\end{proof}

Parts (ii) and (iii) contain the whole mechanism. Inside the exit
continuation region the premium $W$ is strictly $r$-\emph{sub}harmonic:
$v$ is a discounted martingale while the price is a strict discounted
supermartingale, so the premium grows in expectation and delaying purchase
pays. The only exceptions are skew points with $\beta_z>0$. There, (iii)
shows that $W$ has a concave corner of exactly the kind that makes stopping
locally optimal: the $O(\sqrt t)$ local-time loss beats the $O(t)$ drift
gain.

Suppose $\Pplus\ne\emptyset$ and enumerate it as $p_1<\dots<p_k$. Put
\[
y_j:=F(p_j),\qquad h_j:=W(p_j)/\varphi_r(p_j)>0,
\]
where $h_j>0$ by Lemma~\ref{lem:Wlocal}(i),(iv). Let $\kappa$ be the
function on $[0,\infty)$ equal to $h_j$ at $y_j$ and to $0$ elsewhere, and
let $K$ be its least concave majorant on $[0,\infty)$. Define the \emph{hull
set}
\[
\Gamma:=\{p_j:\ K(y_j)=h_j\}.
\]

\begin{lemma}[the hull]
\label{lem:hull}
$K(0)=0$; $K$ is non-decreasing, piecewise linear and strictly positive on
$(0,\infty)$; every breakpoint of $K$ is a point $y_j$ with $p_j\in\Gamma$;
and $K\equiv\max_jh_j$ on $[y_{j^\ast},\infty)$, where
$j^\ast:=\min\{j:\ h_j=\max_i h_i\}$. In particular $p_{j^\ast}\in\Gamma$, so
$\Gamma\neq\emptyset$.
\end{lemma}

\begin{proof}
Write $(y_0,h_0):=(0,0)$ and define
\[
\tilde K(y):=
\begin{cases}
\max\Bigl\{\sum_{i=0}^{j^\ast}\lambda_ih_i:\ \lambda_i\ge0,\
  \sum_{i=0}^{j^\ast}\lambda_i=1,\ \sum_{i=0}^{j^\ast}\lambda_iy_i=y\Bigr\},
  & 0\le y\le y_{j^\ast},\\[4pt]
h_{j^\ast}, & y>y_{j^\ast}.
\end{cases}
\]
On $[0,y_{j^\ast}]$ this is the upper boundary of the convex hull of
$(y_0,h_0),\dots,(y_{j^\ast},h_{j^\ast})$. For each such $y$ the set $P_y$ of
admissible weights is a compact polytope, non-empty because $y$ is a convex
combination of $y_0=0$ and $y_{j^\ast}$, so the maximum is attained. It is
concave: if $y=\theta y'+(1-\theta)y''$, mixing optimal weights for $y'$ and
$y''$ is admissible for $y$, so
$\tilde K(y)\ge\theta\tilde K(y')+(1-\theta)\tilde K(y'')$. Since every
$h_i\le h_{j^\ast}=\tilde K(y_{j^\ast})$, the concave function $\tilde K$
attains its maximum over $[0,y_{j^\ast}]$ at the right endpoint and is
therefore non-decreasing there. Gluing on the constant $h_{j^\ast}$ keeps it
concave and non-decreasing on $[0,\infty)$.

The objective $\sum_i\lambda_ih_i$ is linear and $P_y$ is cut out by two
equality constraints and the constraints $\lambda_i\ge0$, so the maximum is
attained at an extreme point of $P_y$, which has at most two non-zero
weights. Hence $\tilde K(y)$ is the largest value at $y$ of the chords joining
pairs $(y_i,h_i)$, $(y_k,h_k)$ with $y_i\le y\le y_k$, the case $i=k$
included. To see that $\tilde K$ is piecewise linear with breakpoints only at
points $y_j$ where $\tilde K(y_j)=h_j$, let $y\in(0,y_{j^\ast})$ be either
not a data point or a data point $y_j$ with $\tilde K(y_j)>h_j$. Then an
optimal chord at $y$ joins points $y_i<y<y_k$, and it is admissible at every
$y'\in[y_i,y_k]$. Thus $g:=\tilde K-\text{chord}$ is concave and
non-negative on $[y_i,y_k]$ and vanishes at the interior point $y$. A
concave function attaining its minimum at an interior point is constant near
it, since $g(y)\ge\tfrac12\bigl(g(y-\delta)+g(y+\delta)\bigr)\ge g(y)$ forces
equality for small $\delta$. So $\tilde K$ is affine near $y$, and every
breakpoint is a point $y_j$ with $\tilde K(y_j)=h_j$. It majorises $\kappa$: take
$\lambda=\delta^{(j)}$ at $y_j$ for $j\le j^\ast$, where $\delta^{(j)}$
denotes the weights with $\delta^{(j)}_j=1$ and $\delta^{(j)}_i=0$ for $i\ne j$; use $h_j\le h_{j^\ast}$ for
$j>j^\ast$, and note that $\tilde K\ge0$ because every $h_i\ge0$. Finally
$\tilde K(0)=0$, since $y_i>0$ for $i\ge1$ forces $\lambda=\delta^{(0)}$ at $y=0$.

Now let $K'$ be any concave majorant of $\kappa$ on $[0,\infty)$. For admissible
weights at $y\le y_{j^\ast}$, Jensen's inequality gives
\[
K'(y)=K'\Bigl(\sum_i\lambda_iy_i\Bigr)\ge\sum_i\lambda_iK'(y_i)
\ge\sum_i\lambda_ih_i ,
\]
and taking the maximum over the weights gives $K'\ge\tilde K$ on
$[0,y_{j^\ast}]$. Since $K'\ge\kappa\ge0$ and a concave function on a
half-line that is bounded below is non-decreasing, for $y>y_{j^\ast}$ we
have $K'(y)\ge K'(y_{j^\ast})\ge h_{j^\ast}=\tilde K(y)$. Hence $K=\tilde K$,
and every claim is read off $\tilde K$. Positivity on $(0,\infty)$ holds
because $K$ is concave, $K(0)=0$ and $K(y_{j^\ast})>0$.
\end{proof}

\begin{theorem}[entry duality]
\label{thm:duality}
Assume Hypothesis~\ref{hyp:exit}.
\begin{enumerate}
\item[(a)] If $\Pplus=\emptyset$ then $W\equiv0$ and $u\equiv0$: entry is
  never strictly profitable.
\item[(b)] If $\Pplus\neq\emptyset$ then
  \begin{equation}
  u=\varphi_r\cdot K\!\circ\! F,
  \label{eq:uhull}
  \end{equation}
  and $\tau_\Gamma:=\inf\{t\ge0:S_t\in\Gamma\}$ is optimal. Moreover
  $u>W$ on $(0,\infty)\setminus\Gamma$, so the entry region
  $\{u=W\}$ is exactly $\Gamma$. In particular
  \[
  \emptyset\ne\{u=W\}\subset\Pplus .
  \]
\end{enumerate}
\end{theorem}

\begin{proof}
(a) If $\Pplus=\emptyset$ then every $\beta_z<0$, so $g$ is
$r$-superharmonic by \eqref{eq:gfacts}, hence $v=g$ \M[Prop.~8.1], $W\equiv0$
and $u\equiv0$.

(b) Put $\hat U:=\varphi_r\cdot K\circ F$.

\emph{Step 1: $\hat U$ is $r$-harmonic off $\Gamma$.} By Lemma~\ref{lem:hull},
$K$ is affine on each component of $[0,\infty)\setminus F(\Gamma)$, say
$K(y)=A+By$. There $\hat U=A\varphi_r+B\psi_r$, which is $r$-harmonic in
the sense of Section~\ref{sec:setting}, including across any skew points in the
component. On the first component $A=0$, since $K(0)=0$, so
$\hat U=B\psi_r$; on the last, $B=0$ and $\hat U=A\varphi_r$. In
particular $\hat U$ is bounded.

\emph{Step 2: $\Delta_p(\hat U)\le0$ at $p\in\Gamma$.} Let $K$ have slopes
$B_\mp$ on either side of $y=F(p)$, and write $\delta B:=B_+-B_-\le0$ by
concavity. The pieces of $\hat U$ on the two sides differ by the
$r$-harmonic function $\delta A\,\varphi_r+\delta B\,\psi_r$, and
continuity of $K$ at $F(p)$ (Lemma~\ref{lem:hull}) gives $\delta A+\delta B\,F(p)=0$. Extend the left piece
across $p$ as an $r$-harmonic function; its $\Delta_p$ vanishes, so
\[
\Delta_p(\hat U)=(1+\beta_p)\,\delta B\,\bigl(\psi_r'(p^+)-F(p)\varphi_r'(p^+)\bigr)
=(1+\beta_p)\,\delta B\,\frac{\mathcal W(p^+)}{\varphi_r(p)}\ \le\ 0
\]
by \eqref{eq:wronskian}. With Step 1, $\hat U$ is non-negative and
$r$-superharmonic on $(0,\infty)$, hence $r$-excessive.

\emph{Step 3: $\hat U>W$ off $\Gamma$.} Off $\Cex$, $W=0<\hat U$ by
Lemma~\ref{lem:hull}. Let $J$ be a component of $\Cex$ and $D:=\hat U-W$
on $\bar J$. Then:
\begin{itemize}
\item at an endpoint $e$ of $J$, $e\notin\Cex$ because $\Cex$ is open, so
  $W(e)=0$, while $\hat U(e)=\varphi_r(e)\,K(F(e))>0$ by Lemma~\ref{lem:hull}
  since $F(e)>0$; hence $D(e)>0$;
\item at $p_j\in\Pplus\cap J$, $\hat U(p_j)=\varphi_r(p_j)K(y_j)$ since
  $y_j=F(p_j)$, and $W(p_j)=\varphi_r(p_j)h_j$ by the definition of $h_j$, so
  $D(p_j)=\varphi_r(p_j)\bigl(K(y_j)-h_j\bigr)\ge0$ because $K$ majorises
  $\kappa$ and $\kappa(y_j)=h_j$; as $\varphi_r(p_j)>0$, equality holds if
  and only if $K(y_j)=h_j$, that is, $p_j\in\Gamma$;
\item on $J\setminus Z$, $\hat U$ is $r$-harmonic by Step 1 (as $\Gamma\subset Z$),
  so $(\Lop-r)D=-(r-\mu)x<0$ by Lemma~\ref{lem:Wlocal}(ii);
\item at $z\in(Z\setminus\Pplus)\cap J$: since $z\notin\Gamma\subset\Pplus$,
  $K$ is affine on a neighbourhood of $F(z)$ by Lemma~\ref{lem:hull}, so near
  $z$ we have $\hat U=A\varphi_r+B\psi_r$ with the same $A,B$ on both sides of
  $z$ (Step 1); as $\varphi_r$ and $\psi_r$ satisfy the transmission condition
  at every skew point and $\Delta_z$ is linear, $\Delta_z(\hat U)=0$, and so
  $\Delta_z(D)=2\beta_z<0$ by Lemma~\ref{lem:Wlocal}(iii).
\end{itemize}
Suppose $D(x_0)\le0$ for some $x_0\in\bar J\setminus\Gamma$, and let $x_1$
minimise $D$ over $\bar J$, choosing $x_1=x_0$ if $\min D=0$. Then
$D(x_1)\le0$ and $x_1\notin\Gamma$, so by the first two bullets $x_1$ is
either a point of $J\setminus Z$ or a point of $Z\setminus\Pplus$. In the
first case $D'(x_1)=0$ and $D''(x_1)\ge0$, so
$(\Lop-r)D(x_1)\ge-rD(x_1)\ge0$, contradicting the third bullet. In the
second case a minimum requires $D'(x_1^-)\le0\le D'(x_1^+)$, whence
$\Delta_{x_1}(D)\ge0$, contradicting the fourth bullet. Hence $D>0$ on
$\bar J\setminus\Gamma$. On $\Gamma$, $D=0$ by the second bullet.

\emph{Step 4: $\hat U$ is attained by $\tau_\Gamma$.} By Step 1 and
boundedness, $e^{-r(t\wedge\tau_\Gamma)}\hat U(S_{t\wedge\tau_\Gamma})$
is a bounded martingale. On $\{\tau_\Gamma<\infty\}$ it converges to
$e^{-r\tau_\Gamma}W(S_{\tau_\Gamma})$, since $\hat U=W$ on $\Gamma$, and on
$\{\tau_\Gamma=\infty\}$ it converges to $0$. So
$\hat U(x)=\E_x[e^{-r\tau_\Gamma}W(S_{\tau_\Gamma})]\le u(x)$.

\emph{Step 5: $u\le\hat U$.} For any stopping time $\tau$, Steps 2 and 3 and
optional sampling for the non-negative supermartingale
$e^{-rt}\hat U(S_t)$ give
$\E_x[e^{-r\tau}W(S_\tau)\mathbf 1_{\{\tau<\infty\}}]\le
\E_x[e^{-r\tau}\hat U(S_\tau)\mathbf 1_{\{\tau<\infty\}}]\le\hat U(x)$.

Hence $u=\hat U$, $\tau_\Gamma$ is optimal, and $\{u=W\}=\Gamma$ by Step 3.
\end{proof}

\begin{remark}[what was not needed]
\label{rem:notneeded}
The proof uses no information about the shape of $W$ beyond
Lemma~\ref{lem:Wlocal}. It needs no smooth fit for $u$, and no smooth-fit
analysis of the non-smooth reward is required: whether $W$ is $C^1$ or only
$C^0$ at the edges of $\Cex$ is irrelevant, because $W=0<u$ there. The
concavity of $K$ is the Dayanik--Karatzas characterisation
\citep{DayanikKaratzas03}. What is specific to this problem is that the
majorant touches only at finitely many points, which are known in advance.
\end{remark}

\begin{remark}[robustness to transaction costs]
\label{rem:costs}
Replace the reward by $W_c:=(1-\lambda')v-(1+\lambda)g-c$, with
proportional costs $\lambda,\lambda'\in[0,1)$ and a fixed cost $c\ge0$, and
let $u_c(x):=\sup_\tau\E_x[e^{-r\tau}W_c(S_\tau)]$ with the same convention
on $\{\tau=\infty\}$, so that $u_c\ge0$. The exit value with proceeds
$(1-\lambda')x$ is $(1-\lambda')v$ by linearity, with the same continuation
region $\Cex$. We claim that
\[
\{u_c=W_c\}\subset\Pplus .
\]
The proof is local; the global argument of Theorem~\ref{thm:duality} is not
available, because Lemma~\ref{lem:hull} needs every $h_j>0$, whereas $W_c$
may be negative even at support points. It suffices to show $u_c(x)>W_c(x)$
for every $x\notin\Pplus$, and there are three cases.
\begin{itemize}
\item $x\in\Cex\setminus Z$. There $(\Lop-r)v=0$, $(\Lop-r)g=(\mu-r)x$ and
  $(\Lop-r)c=-rc$, so $(\Lop-r)W_c=(1+\lambda)(r-\mu)x+rc>0$. Let $I\ni x$ be
  a compact interval in $\Cex\setminus Z$ with $x$ in its interior, and let
  $\tau$ be the exit time of $I$, which is a.s.\ finite and strictly positive.
  By Dynkin's formula,
  $\E_x[e^{-r\tau}W_c(S_\tau)]=W_c(x)+\E_x\int_0^\tau e^{-rs}(\Lop-r)W_c(S_s)\,ds
  >W_c(x)$, so $u_c(x)>W_c(x)$.
\item $x=z\in Z\cap\Cex$ with $\beta_z<0$. Take $I$ as before but
  containing no skew level other than $z$. By Lemma~\ref{lem:ito} and optional
  sampling ($W_c$ is bounded on $I$ and $\E_z\tau<\infty$),
  \[
  \E_z\bigl[e^{-r\tau}W_c(S_\tau)\bigr]-W_c(z)
  =\E_z\!\int_0^\tau e^{-rs}(\Lop-r)W_c(S_s)\,ds
  +\tfrac12\Delta_z(W_c)\,\E_z\!\int_0^\tau e^{-rs}\,d\ell^z_s .
  \]
  Both terms are positive: the first as in the previous case, since $S$
  spends zero time at $z$; the second because $\Delta_z(v)=0$ by
  Hypothesis~\ref{hyp:exit} and $\Delta_z(g)=2\beta_z$, so that
  $\Delta_z(W_c)=-2(1+\lambda)\beta_z>0$, while
  $\E_z\int_0^\tau e^{-rs}\,d\ell^z_s>0$ because a path started at $z$
  accumulates local time there immediately. Hence again $u_c(z)>W_c(z)$.
\item $x\notin\Cex$. Then $v(x)=g(x)$, so $W_c(x)=-(\lambda+\lambda')x-c$.
  If at least one of $\lambda,\lambda',c$ is positive, $W_c(x)<0\le u_c(x)$. If there are no costs,
  $W_c=W$ and $W(x)=0<u(x)$ as in Theorem~\ref{thm:duality}.
\end{itemize}
What does not carry over is non-emptiness and the hull formula: only
supports with $W_c(p)>0$ can enter, and the entry region may be empty, in
which case never buying is optimal. Whether $u_c$ is given by the hull over
$\{p\in\Pplus:W_c(p)>0\}$, with $u_c\equiv0$ when $W_c\le0$ on $\Pplus$, we
leave open: Step~3 of the proof of Theorem~\ref{thm:duality} uses $W\ge0$
and would have to be redone for a reward that can be negative.
\end{remark}

\section{Support and resistance: buy at support}
\label{sec:main}

Throughout this section $n=2$, $z_1=L<z_2=H$, $\beta_L\in(0,1)$ and
$\beta_H\in(-1,0)$, as in \M[(1)]. Let $(a,b)$ be the exit continuation band
in whichever regime obtains \M[\S4], and write $G:=G_L$ for the function
\eqref{eq:Gz} at $z=L$. By \M[Thm.~4.1], $G$ is strictly convex with unique
minimum $G(L)=L$.

\begin{corollary}[buy exactly at support]
\label{cor:main}
In each of regimes A, B and C the entry region is $\{L\}$, the optimal entry
time is $\tau_L$, and
\begin{equation}
u(x)=\bigl(G(a)-L\bigr)\,\E_x\bigl[e^{-r\tau_L}\bigr]
=\bigl(G(a)-L\bigr)\times
\begin{cases}
(x/L)^{\alpha_2}, & x\le L,\\[2pt]
\varphi_r(x)/\varphi_r(L), & x\ge L,
\end{cases}
\label{eq:umain}
\end{equation}
where $\varphi_r=x^{\alpha_1}$ on $[H,\infty)$ and $\varphi_r$ has
coefficient pair $T_{\beta_H}(H)(1,0)^{\!\top}$ on $[L,H]$. Moreover $u>W$
on $(0,\infty)\setminus\{L\}$.
\end{corollary}

\begin{proof}
Here $\Pplus=\{L\}$, so $k=1$, and a single point is on its own hull:
$\Gamma=\{L\}$. By \eqref{eq:uhull}, $u=h_1\varphi_r\cdot(F/y_1)$ on
$(0,L]$ and $u=h_1\varphi_r$ on $[L,\infty)$; that is,
$u=W(L)\psi_r/\psi_r(L)$ and $u=W(L)\varphi_r/\varphi_r(L)$, which by
Lemma~\ref{lem:hitting} is $W(L)\E_x[e^{-r\tau_L}]$. In every regime
$a<L$, and on $(a,L)$ the exit value has coefficient pair $\vk(a)$ by smooth
fit (Lemma~\ref{lem:smoothfit}). So $v(L)=e_L\vk(a)=G(a)$ and $W(L)=G(a)-L$. Finally
$\psi_r=x^{\alpha_2}$ on $(0,L]$ by definition.
\end{proof}

\begin{remark}[a corner, not a free boundary]
\label{rem:corner}
No boundary is optimised and there is no smooth fit. From \eqref{eq:umain},
$u'(L^-)=\alpha_2W(L)/L>0$ and $u'(L^+)=W(L)\varphi_r'(L^+)/\varphi_r(L)<0$,
so $\Delta_L(u)<0$ strictly: $u$ has a concave corner at $L$, and $L$ is an
isolated point of the stopping set. This is the entry counterpart of
\M[Thm.~4.3], where the exit boundary pins at the kink $H$. There pinning
depends on the criterion \M[(23)]; here it happens in every regime. The
reason is that the reward itself has the corner:
$\Delta_L(W)=-2\beta_L<0$ (Lemma~\ref{lem:Wlocal}(iii); cf.\ the concave
corner of $W$ at $L$ noted at the end of the proof of \M[Prop.~5.1]).
That smooth fit can fail at points where the scale or speed measure is
singular is known in general \citep{SalminenTa15}.
\end{remark}

\begin{remark}[stopping at a skew point]
\label{rem:AS}
For skew Brownian motion with an upward skew, \citet[Prop.~1]{AlvarezSalminen17}
show that the skew point lies in the continuation region whenever the
pay-off is non-decreasing (their standing assumption (g1)) and satisfies
$0\le g'(0^-)<g'(0^+)$ or $0<g'(0^-)\le g'(0^+)$. Corollary~\ref{cor:main}
stops exactly at an upward skew point, and the two are consistent: the entry
reward $W$ is not non-decreasing at $L$, but has a concave corner there, with
$W'(L^-)>0>W'(L^+)$ (end of the proof of \M[Prop.~5.1]). Their argument
gives a necessary condition for a skew point to be a stopping point
\citep[Rem.~1(2)]{AlvarezSalminen17}. Matching their scale function with
ours, their skewness parameter is $(1+\beta_z)/2$, and the condition reads
$(1+\beta_z)R'(z^+)\le(1-\beta_z)R'(z^-)$ for a reward $R$, that is,
$\Delta_z(R)\le0$; this is the condition that case~2 of
Remark~\ref{rem:costs} shows to be necessary. For the seller at $L$,
$\Delta_L(g)=2\beta_L>0$, so $L$ cannot be a selling point; for the buyer,
$\Delta_L(W)=-2\beta_L<0$, and $L$ is a buying point. The signs that keep the
seller waiting at a support are exactly those that make the buyer stop
there; this is the duality of Theorem~\ref{thm:duality} seen locally.
Isolated stopping points at a skew point are not new in themselves:
Presman and Yurinsky (2012), as reported in
\citet[Rem.~1(3)]{AlvarezSalminen17}, find one for skew geometric Brownian
motion with a downward skew and a piecewise linear pay-off. What is new here
is that every point of the entry region is of this kind, and that the reward
producing them is the exit premium itself.
\end{remark}

\begin{remark}[the round trip]
\label{rem:roundtrip}
Since $L\in(a,b)$, the investor who buys at $\tau_L$ is immediately inside
the exit band. The complete strategy is: wait for the price to hit $L$; buy;
sell at the first exit from $(a,b)$, which means a stop-loss at $a$ or a sale
at $b$, with $b=H$ in regime~C. From $x\le L$ its value is
$(G(a)-L)(x/L)^{\alpha_2}$, which depends on the exit problem only through
the stop-loss level $a$.
\end{remark}

\begin{remark}[dependence on $\beta_H$]
\label{rem:betaH}
The rule $\tau_L$ depends on no parameter except $L$. The value depends on
$\beta_H$ in two ways. First, through $a$, which varies with $\beta_H$ only
in regime~B: in regime~A, $a$ is the $L$-band endpoint, and in regime~C it
solves \M[(22)], in which $\beta_H$ does not appear. Second, through
$\varphi_r$ on $(L,\infty)$, which gives the discounted expectation
$\E_x[e^{-r\tau_L}]$ of the time for the price to fall back to $L$; it
reflects the price dynamics, not the buying rule or the selling rule. This is consistent with the fact that $v$
itself depends on $\beta_H$ only through $b$: on $(a,L)$ and $(L,H)$ its
coefficient pairs involve only $a$ and $\beta_L$, on $(H,b)$ the pair is
$\vk(b)$ by smooth fit, and in regime~B continuity at $H$, which does not
involve $\beta_H$, determines $a$ from $b$ \M[\S4.3]. Entry decisions
therefore carry no identifying information about $\beta_H$ beyond what the
exit band already carries.
\end{remark}

\section{Two supports: a three-way classification}
\label{sec:samesign}

Now let $\beta_L,\beta_H\in(0,1)$, the case of \M[\S7], so that $\Pplus=\{L,H\}$
and both points lie in $\Cex$. The exit region there consists either of two
disjoint bands, one around each level, or of one merged band.

\begin{corollary}[which support to buy at]
\label{cor:samesign}
Let $\rho:=W(L)/W(H)$ and
\[
\underline\rho:=\frac{\psi_r(L)}{\psi_r(H)}=\E_L\bigl[e^{-r\tau_H}\bigr]\in(0,1),
\qquad
\overline\rho:=\frac{\varphi_r(L)}{\varphi_r(H)}=\frac{1}{\E_H\bigl[e^{-r\tau_L}\bigr]}\in(1,\infty).
\]
Then
\[
\{u=W\}=
\begin{cases}
\{H\}, & \rho<\underline\rho,\\
\{L,H\}, & \underline\rho\le\rho\le\overline\rho,\\
\{L\}, & \rho>\overline\rho.
\end{cases}
\]
Equivalently, $L$ is an entry point if and only if
$W(L)\ge\E_L[e^{-r\tau_H}]\,W(H)$, and $H$ is an entry point if and only if
$W(H)\ge\E_H[e^{-r\tau_L}]\,W(L)$: one buys at a support unless the other
support is worth more after discounting for the time it takes to get there.
In the case $\{L,H\}$,
\[
u(x)=\begin{cases}
W(L)\,\psi_r(x)/\psi_r(L), & x\le L,\\
A\varphi_r(x)+B\psi_r(x)\ \text{with } u(L)=W(L),\ u(H)=W(H), & L\le x\le H,\\
W(H)\,\varphi_r(x)/\varphi_r(H), & x\ge H .
\end{cases}
\]
\end{corollary}

\begin{proof}
With two points $y_L<y_H$ in Theorem~\ref{thm:duality}, the hull through the
origin passes through $(y_L,h_L)$ if and only if $h_L/y_L\ge h_H/y_H$.
Otherwise $(y_L,h_L)$ lies strictly below the chord from the origin to
$(y_H,h_H)$. The hull passes through $(y_H,h_H)$ if and only if
$h_H\ge h_L$; otherwise $(y_H,h_H)$ lies below the flat part at height
$h_L$. Since $h/y=W/\psi_r$ and $h=W/\varphi_r$, these conditions are
$W(L)/\psi_r(L)\ge W(H)/\psi_r(H)$ and $W(H)/\varphi_r(H)\ge
W(L)/\varphi_r(L)$, which are the stated inequalities. The expectations come
from Lemma~\ref{lem:hitting}, and $\underline\rho<1<\overline\rho$ because
$\psi_r$ is increasing and $\varphi_r$ decreasing.
\end{proof}

Since the window $[\underline\rho,\overline\rho]$ contains $1$ in its
interior, buying at both supports is optimal whenever $W(L)=W(H)$. Moreover
$\rho$, $\underline\rho$ and $\overline\rho$ depend continuously on the
parameters, so the strict inequalities $\underline\rho<\rho<\overline\rho$
persist under small perturbations: buying at both supports is optimal on an
open set of parameters, not only in a knife-edge case.

\begin{remark}[closed forms for the premia]
\label{rem:premia}
The premia $W(L)$ and $W(H)$ entering Corollary~\ref{cor:samesign} are
explicit.
\begin{itemize}
\item \emph{Disjoint bands.} Each band is the single-skew band of its own
  level. By the scale invariance of $g(x)=x$ under geometric Brownian motion,
  the single-skew problem at $z$ is $z$ times the problem at $1$. Hence
  $W(z)=z\,w(\beta_z)$, where $w(\beta):=G_1(\bar a(\beta))-1$ is the
  premium at $1$ of the unit single-skew problem (with $G_1$ the function $G$
  at $L=1$), and $\rho=(L/H)\,w(\beta_L)/w(\beta_H)$.
\item \emph{Merged band $(a,b)$.} The coefficient pairs on $(a,L)$ and
  $(H,b)$ are $\vk(a)$ and $\vk(b)$, so $W(L)=G(a)-L$ and $W(H)=G_H(b)-H$,
  with $G_H$ as in \eqref{eq:Gz}.
\end{itemize}
By Lemma~\ref{lem:wmono} below, $w$ is strictly increasing, so in the
disjoint-band case $\rho$ increases with $\beta_L$ and decreases with
$\beta_H$. The effect on the entry region also depends on the window
$[\underline\rho,\overline\rho]$, which moves with $\beta_L$ and $\beta_H$
through $\psi_r$ and $\varphi_r$; we do not pursue this here.
\end{remark}

\begin{lemma}[the single-skew premium is increasing]
\label{lem:wmono}
The premium $w$ of Remark~\ref{rem:premia} is strictly increasing on $(0,1)$.
\end{lemma}

\begin{proof}
Take $L=1$ and write $G:=G_1$, so $G(m)=\kappa_1m^{1-\alpha_1}+\kappa_2m^{1-\alpha_2}$,
and $D(m):=\kappa_1\alpha_1m^{1-\alpha_1}+\kappa_2\alpha_2m^{1-\alpha_2}$.
By \M[Thm.~4.1] the single-skew band $(a,b)$ with skew $\beta\in(0,1)$ solves
\[
G(a)=G(b)=:v,\qquad \pi(\beta):=\frac{1+\beta}{1-\beta}=\frac{D(a)}{D(b)},
\]
and $G$ is strictly convex with minimum $G(1)=1$, so $a<1<b$ and
$w(\beta)=G(a)-1=v-1$. It therefore suffices to show that $v$ is strictly
increasing in $\beta$. Since
$\kappa_1(1-\alpha_1)m^{1-\alpha_1}+\kappa_2(1-\alpha_2)m^{1-\alpha_2}=mG'(m)$,
we have $D=G-mG'$, so on the level set
\[
D(a)=v-aG'(a),\qquad D(b)=v-bG'(b).
\]
Parametrise the level set by $v>1$. Then $a(v)$ decreases and $b(v)$
increases, with $a'(v)=1/G'(a)$ and $b'(v)=1/G'(b)$, and
\[
\frac{d}{dv}D(a(v))=-\frac{aG''(a)}{G'(a)}>0,\qquad
\frac{d}{dv}D(b(v))=-\frac{bG''(b)}{G'(b)}<0,
\]
because $G''>0$, $G'(a)<0$ and $G'(b)>0$. Here $D(a)=v-aG'(a)>0$ for every
$v$, and at any solution the second equation gives
$D(b)=(1-\beta)D(a)/(1+\beta)>0$, so solutions can only occur where
$D(b(v))>0$. That set is an initial interval: $D(b(v))$ decreases from
$D(1)=\kappa_1\alpha_1+\kappa_2\alpha_2=1$ at $v=1$ (where $a=b=1$), and it
becomes negative for large $v$, since then
$D(b)\sim(\kappa_1-\kappa_1(1-\alpha_1))b^{1-\alpha_1}=\kappa_1\alpha_1b^{1-\alpha_1}$.
So $D(b(v))>0$ exactly on some $[1,v^\ast)$, and there
\[
R(v):=\frac{D(a(v))}{D(b(v))},\qquad
R'(v)=\frac{\frac{d}{dv}D(a(v))\,D(b(v))-D(a(v))\,\frac{d}{dv}D(b(v))}{D(b(v))^2}>0,
\]
both terms of the numerator being positive. Thus $R$ increases strictly on
$[1,v^\ast)$, from $R(1)=1$ to $+\infty$. Since $\pi$ is a strictly increasing
bijection from $(0,1)$ onto $(1,\infty)$ and $\pi(\beta)=R(v)$, each
$\beta\in(0,1)$ determines a unique $v\in(1,v^\ast)$, and $v$, hence
$w(\beta)=v-1$, is strictly increasing in $\beta$.
\end{proof}

\section{The remaining sign pairs}
\label{sec:othersigns}

With $G_H$ as in \eqref{eq:Gz}, we have the following.

\begin{corollary}[the other sign pairs]
\label{cor:othersigns}
Let $n=2$, $z_1=L<z_2=H$.
\begin{enumerate}
\item[(a)] If $\beta_L,\beta_H<0$, then $W\equiv0$ and $u\equiv0$: entry is
  never strictly profitable.
\item[(b)] If $\beta_L<0<\beta_H$, assume Hypothesis~\ref{hyp:exit} and that
  the component of $\Cex$ containing $H$ is an interval $(a,b)$ whose upper
  end $b>H$ is a smooth-fit boundary, as in the geometries of
  \M[Prop.~8.2 and Rem.~8.3]. Then the entry region is $\{H\}$ and
  \[
  u(x)=\bigl(G_H(b)-H\bigr)\,\E_x\bigl[e^{-r\tau_H}\bigr].
  \]
  In the geometry of \M[Prop.~8.2], where $(a,b)$ is the single-skew band at
  $H$, $G_H(b)-H=H\,w(\beta_H)$ with $w$ as in Remark~\ref{rem:premia}.
\end{enumerate}
\end{corollary}

\begin{proof}
(a) is Theorem~\ref{thm:duality}(a). For (b), $\Pplus=\{H\}$, so $k=1$ and
$\Gamma=\{H\}$, and as in the proof of Corollary~\ref{cor:main},
\eqref{eq:uhull} and Lemma~\ref{lem:hitting} give
$u=W(H)\,\E_x[e^{-r\tau_H}]$. By Lemma~\ref{lem:Wlocal}(iv), $H\in\Cex$, so
$a<H<b$. The interval $(H,b)$ contains no skew point, and smooth fit at $b$
fixes the coefficient pair of $v$ there as $\vk(b)$
(Lemma~\ref{lem:smoothfit}); hence
$v(H)=e_H\vk(b)=G_H(b)$ and $W(H)=G_H(b)-H$. The last claim follows
from the scale invariance used in Remark~\ref{rem:premia}.
\end{proof}

In particular the lower level $L$, which pushes the price down, never
attracts entry, whether or not the exit band reaches it, and even when the
exit band pins at $a=L$ as in \M[Rem.~8.3].

\section{Conclusion}
\label{sec:disc}

For a geometric multi-skew Brownian motion, the decision to buy a share that
will then be sold optimally is governed by the exit premium $W=v-g$. Inside
the exit continuation region this premium is strictly $r$-subharmonic, so
buying is never optimal there, except at the levels where the price receives
an upward local-time push. The entry problem therefore reduces to a finite
one (Theorem~\ref{thm:duality}): its value is a concave hull through
finitely many points, and its entry region is a set of support levels. With
one support and one resistance, one buys exactly at the support, in every
exit regime (Corollary~\ref{cor:main}); with two supports, a closed-form
window decides between the two levels and both
(Corollary~\ref{cor:samesign}). This contrasts with the entry intervals of the
buy-low, sell-high literature and of the regime-switching model of
\citet{Henderson26}.

\paragraph{Open directions.}
\begin{enumerate}
\item[(i)] Transaction costs (Remark~\ref{rem:costs}). The hull with a
  possibly negative reward, and the threshold cost above which a support
  drops out of the entry region.
\item[(ii)] Repeated round trips. Each cycle buys at a point of $\Gamma$ and
  sells on leaving an exit band, which gives an optimal switching problem
  whose entry set is presumably still contained in $\Pplus$; for a
  Cox--Ingersoll--Ross process, \citet{LeungLiWang14} find that the switching
  thresholds lie strictly inside those of a single round trip.
\item[(iii)] Many supports. The combinatorics of the hull in
  Theorem~\ref{thm:duality} as $n$ grows, and the comparative statics of
  $\Gamma$ in the $\beta_i$, starting with how the window of
  Corollary~\ref{cor:samesign} moves with $\beta_L$ and $\beta_H$.
\end{enumerate}

\section*{Declaration on the use of generative AI}
The author used Anthropic's Claude as an assistant in this work, including
for generating and developing ideas, for drafting and checking proofs, for
numerical verification, for the literature search, and for preparing the
manuscript. The author checked all definitions, statements and proofs,
reviewed and edited the output, and takes full responsibility for the
content of the paper.

\appendix
\section{Numerical check}
\label{app:num}

We computed $v$ and $u$ by brute force,
as least concave majorants of $g/\varphi_r$ and $W/\varphi_r$ in the
$F$-coordinate on a grid of $4\times10^5$ geometrically spaced points in
$[10^{-4},10^{4}]$; the computation assumes none of the results above.
Parameters: $r=0.05$, $\mu=0.01$, $\sigma^2=0.04$, $L=1$.

\begin{table}[h]
\centering
\small
\begin{tabular}{rrcccccc}
\toprule
$H$ & $\beta_H$ & regime & $a$ & $b$ & $W(L)$ (brute force) & $G(a)-L$ &
contact set\\
\midrule
1.60 & $-0.30$ & A & 0.6253 & 1.4654 & 0.182971 & 0.182971 & $\{1\}$\\
1.15 & $-0.20$ & B & 0.6683 & 1.3422 & 0.137656 & 0.137656 & $\{1\}$\\
1.15 & $-0.50$ & C & 0.7003 & 1.1500 & 0.109221 & 0.109221 & $\{1\}$\\
\bottomrule
\end{tabular}
\caption{Main case, $\beta_L=0.9$. On $[0.3,3]$ the brute-force $u$ agrees
with \eqref{eq:umain} to within $3\times10^{-10}$ (grid interpolation), and
$u-W\ge0$ with equality only at $L$.}
\end{table}

\begin{table}[h]
\centering
\small
\begin{tabular}{rrrrrrccc}
\toprule
$H$ & $\beta_L$ & $\beta_H$ & $W(L)$ & $W(H)$ & $\rho$ &
$[\underline\rho,\overline\rho]$ & rule & contact set\\
\midrule
3.0 & 0.9 & 0.9 & 0.18297 & 0.54891 & 0.3333 & $[0.2794,\,36.91]$ & $\{L,H\}$ & $\{1,3\}$\\
3.0 & 0.5 & 0.9 & 0.06028 & 0.54891 & 0.1098 & $[0.2091,\,36.91]$ & $\{H\}$ & $\{3\}$\\
3.0 & 0.9 & 0.5 & 0.18297 & 0.18083 & 1.0118 & $[0.2794,\,8.022]$ & $\{L,H\}$ & $\{1,3\}$\\
1.5 & 0.9 & 0.9 & 0.41030 & 0.29383 & 1.3964 & $[0.7845,\,11.28]$ & $\{L,H\}$ & $\{1,1.5\}$\\
1.5 & 0.6 & 0.9 & 0.25058 & 0.28616 & 0.8757 & $[0.6895,\,11.28]$ & $\{L,H\}$ & $\{1,1.5\}$\\
1.3 & 0.9 & 0.3 & 0.25435 & 0.10058 & 2.5287 & $[0.8934,\,1.718]$ & $\{L\}$ & $\{1\}$\\
2.2 & 0.7 & 0.8 & 0.11463 & 0.32380 & 0.3540 & $[0.4135,\,11.91]$ & $\{H\}$ & $\{2.2\}$\\
\bottomrule
\end{tabular}
\caption{Same-sign case: the window rule of Corollary~\ref{cor:samesign}
against the brute-force contact set. All three outcomes occur. The rows with
$H=3$ are disjoint bands, and $W(3)=3\,W(1)$ at equal $\beta$, as
Remark~\ref{rem:premia} predicts.}
\end{table}

\begin{table}[h]
\centering
\small
\begin{tabular}{rrrccc}
\toprule
$H$ & $\beta_L$ & $\beta_H$ & exit band & $W(H)$ & contact set\\
\midrule
1.5 & $-0.5$ & $-0.5$ & none & $0$ & none\\
1.3 & $-0.9$ & $-0.3$ & none & $0$ & none\\
3.0 & $-0.5$ & $0.9$ & $(1.876,\,4.396)$ & 0.548912 & $\{3\}$\\
1.2 & $-0.2$ & $0.9$ & $(0.889,\,1.750)$ & 0.213649 & $\{1.2\}$\\
1.2 & $-0.6$ & $0.9$ & $(1.000,\,1.749)$ & 0.212365 & $\{1.2\}$\\
\bottomrule
\end{tabular}
\caption{The other sign pairs (Corollary~\ref{cor:othersigns}). For
$\beta_L,\beta_H<0$ the brute-force $W$ and $u$ vanish to within
$2\times10^{-12}$. For $\beta_L<0<\beta_H$ the contact set is $\{H\}$ in
all three exit geometries (band above $L$, band straddling $L$, band pinned
at $a=L$), and $u$ agrees with $W(H)\,\E_x[e^{-r\tau_H}]$ to within
$10^{-11}$ on $[0.3,6]$. In the first of these rows the exact value
$G_H(b)-H=H\,w(0.9)=0.548912$ agrees with the brute-force $W(H)$.}
\end{table}

Consistently with Lemma~\ref{lem:wmono}, the same parameters give the
single-skew premium $w(\beta)$ of Remark~\ref{rem:premia} strictly increasing
on the grid $\beta=0.05,0.10,\dots,0.95$, from $6.2\times10^{-4}$ to $0.202$.


\end{document}